\documentclass[a4paper,12pt]{article}

\usepackage{amssymb}
\usepackage{latexsym}
\usepackage{amsmath}
\usepackage{amsthm}
\usepackage{bm}
\usepackage{graphicx}
\usepackage{hyperref}
\usepackage[margin=1.25in]{geometry}
\usepackage{authblk}
\newcommand{\pv}{pos}
\newcommand{\pvi}{\mathit{posI}}
\newcommand{\pve}{\mathit{posE}}
\newcommand{\pvf}{\mathit{posF}}
\usepackage{graphicx}
\usepackage{caption}
\usepackage{tikz}
\usepackage{color}
\usepackage{cases}
\usepackage{enumerate}
\usepackage[inline,shortlabels]{enumitem}

\usepackage{tikz-network}  
\newtheorem{remark}{Remark}

\newtheorem{lemma}{Lemma}

\newtheorem{proposition}{Proposition}
\newtheorem{theorem}{Theorem}

\begin{document}

\title{Peer Selection with Friends and Enemies\footnote{Marcin Dziubi\'{n}ski's work was supported by the Polish National Science Centre through grant 2021/42/E/HS4/00196.}}
\author{Francis Bloch\footnote{Universit\'{e} Paris 1 and Paris School of Economics, 48
Boulevard Jourdan, 75014 Paris, France.
\texttt{francis.bloch@univ-paris1.fr}},   Bhaskar Dutta\footnote{University of Warwick and Ashoka University, CV4 7AL, Coventry, UK \texttt{b.dutta@warwick.ac.uk}} and Marcin Dziubi\'{n}ski\footnote{Institute of Informatics, University of Warsaw, Banacha 2, 02-097, Warsaw, Poland   \texttt{m.dziubinski@mimuw.edu.pl}}}

\maketitle

\begin{abstract}
\small

A planner wants to select one agent out of n agents on the basis of a binary characteristic that is commonly known to all agents but is not observed by the planner. Any pair of agents can either be friends or enemies or impartials of each other. An individual's most preferred outcome is that she be selected. If she is not selected, then she would prefer that a friend be selected, and if neither she herself or a friend is selected, then she would prefer that an impartial agent be selected. Finally, her least preferred outcome is that an enemy be selected.
The planner wants to design a dominant strategy incentive compatible  mechanism in order to be able choose a desirable agent. We derive sufficient conditions for existence of efficient and DSIC mechanisms  when the planner knows the  bilateral relationships between agents. We also show that if the planner does not know  these relationships, then there is no efficient and DSIC mechanism and we compare the relative efficiency of two ``second-best'' DSIC mechanisms. Finally, we obtain sharp characterization results when the network of friends and enemies satisfies structural balance. 
\normalsize
\end{abstract}


\textbf{Declarations of interest:} none


\newpage
\baselineskip=1.45\baselineskip

\section{Introduction}

We consider a selection problem where the planner has to select one individual out of a group of $n$ individuals on the basis of a binary characteristic. For instance, an agent could be either \lq\lq needy" or \lq\lq not needy",  \lq\lq high" or \lq\lq low"  ability, \lq\lq deserving" or \lq\lq not  deserving".   Depending on the context, the chosen agent will then be given a cash grant, or promotion or job offer. To fix ideas, we will henceforth 	denote  the binary characteristic to be one of being needy or not needy. Moreover, the agent selected by the planner will be given a cash grant.

A mechanism design problem arises because the planner does not observe  the neediness status of any agent.  However, all agents know each others' neediness status. This is of course also true in the classical peer selection problem discussed originally by Alon et al \cite{alon2011sum} and Holzman and Moulin \cite{holzman2013impartial}.   However, our framework, while related, differs in important ways from  the Alon-Holzman-Moulin classical framework. 

The planner's problem in our setting is considerably more complicated because all pairs of agents  are also related to each other by bilateral relations of \emph{friendship} and \emph{enmity}.  That is, each pair $i, j$ are either \textit{friends}, \textit{enemies} or \textit{impartial}. These relationships give rise to networks of friends, enemies or impartials.  These bilateral relationships influence preferences. Each agent strictly prefers being selected. If not selected herself, then she would like a friend to be selected regardless of the friend's binary characteristic.  Moreover, she strictly prefers that no enemy be selected, again irrespective of the enemy's neediness status.  Obviously,  any incentive  compatible (probabilistic) peer selection mechanism must ensure that   no agent is able to influence his or her own probability of being selected.  In our setting, incentive compatibility imposes an additional requirement -- agents cannot influence the probability of their friends and enemies from being selected.

We focus on the possibility of constructing \textit{probabilistic mechanisms} under which (i) revealing truthful information will be a weakly dominant strategy for each agent, and (ii) the support of the probability distribution corresponding to truthful reports will be contained in the set of needy agents when the latter set is nonempty.\footnote{The latter is the requirement    of efficiency --  if a needy agent exists, then only a needy agent should be selected.} In contrast, the classical Holzman-Moulin-Alon framework focussed on \textit{deterministic} mechanisms.
We consider two informational scenarios. In the first, the planner knows either the friendship, enemy or impartial relationships of every agent. We derive different sufficient conditions  which allow  the planner to construct three probabilistic  mechanisms  satisfying the two requirements.\footnote{Note that Holzman and Moulin prove that in the deterministic framework, there is \textit{no} efficient mechanism that will induce agents to tell the truth as a dominant  strategy. The difference in results illustrates the difference in frameworks.} 
We have not been able to show whether these sufficient conditions characterise the class of all mechanisms satisfying the two requirements.

We then go onto consider the case when the planner does not know the bilateral relationships of friends and enemies. We first show that in this setting, there is no incentive compatible mechanism that is efficient in the sense of  ensuring that some needy agent is picked whenever there is at least one needy agent. Although there are differences between our framework and   the probabilistic framework of Gibbard \cite{gibbard1977manipulation}, we show that adaptations of both  \textit{random dictatorships} and \textit{duples} are DSIC mechanisms.   We compare these two mechanisms in terms of their relative efficiency.  We show that if every individual has at least two friends, then the random dictatorship is at least as efficient as the duple mechanism. Otherwise, under plausible conditions on the bilateral relationships, the duple mechanism is better than the random dictatorship on a worst-case basis. Both mechanisms are better than the constant mechanism that picks an agent at random.

Finally, we consider networks which satisfy \textit{structural balance} -- for which the enemy of an enemy is a friend, and the friend of an enemy an enemy.\footnote{Structural balance was introduced by Heider \cite{heider1946attitudes} in social psychology, and formalized by Cartwright and Harary \cite{cartwright1956structural} } When structural balance is satisfied we are able to provide a necessary and sufficient condition for existence of an efficient, dominant strategy incentive compatible mechanism. We also compute a lower bound on efficiency of the duples mechanism when the planner does not know the network.

 \section{Related Literature}

 Our paper relates to several strands of the literature. It is most closely related to the large literature in economics and computer science on peer selection mechanisms.   Holzman and Moulin \cite{holzman2013impartial}  provide an axiomatic analysis of  \textbf{deterministic} ``impartial''  or dominant strategy incentive compatible voting rules when individuals nominate a single individual for office. 
  Alon et al. \cite{alon2011sum} consider the problem of designing a dominant strategy mechanism to select $k$ individuals from a group of peers.  They show that  no deterministic efficient  dominant strategy mechanism exists, and then go on to construct approximately efficient, stochastic, dominant strategy mechanisms.\footnote{See also Tamura and Ohseto  \cite{tamura2014impartial}, Berga and Gjorjiev \cite{berga2014impartial}, Mackenzie  \cite{mackenzie2019axiomatic}, Fischer and Klimm \cite{fischer2015optimal}, Bousquet, Norin and Vetta \cite{bousquet2014near}, Kurokawa et al. \cite{kurokawa2015impartial} and Aziz et al. \cite{aziz2016strategyproof} and \cite{aziz2019strategyproof}, Bjelde, Fischer and Klimm \cite{bjelde2017impartial} and Babichenko, Dean and Tenneholtz \cite{babichenko2020incentive}.
}

The peer selection problem  has also been studied in the context of social networks, where the social network is a network of observation, describing who can observe whom in the network.  Bloch and Olckers \cite{bloch2022friend} characterize network structures for which efficient, ex post incentive compatible mechanisms exist. \footnote{See also Baumann \cite{baumann2018identifying}. }

 Faced with the impossibility of constructing an efficient, dominant strategy mechanism, Ben Porath, Dekel and Lipman  \cite{ben2014optimal}  propose to add another instrument for the planner, and allow her to check the report of the agents at a cost.  Mylonanov and Zapechelnuyk \cite{mylovanov2017optimal} analyze the problem when verification can only happen ex post, after the object has been allocated. Kattwinkel \cite{kattwinkel2023costless}, Pereyra and Silva  \cite{pereyra2023optimal} and Bloch, Dutta and Dziubi\'{n}ski \cite{bloch2023selecting} analyze situations where, in addition to the agent's reports, the planner can use a correlated signal on the agents' types, and show how the planner can leverage this correlation to construct efficient Bayesian incentive compatible mechanisms.
 
The entire literature on peer selection assumes that players only care about being selected, and that there are no externalities in players' utilities. Our paper thus departs from the existing literature since players also care about who is selected if they themselves are not.  This is related to the biased juror problem where jurors may deviate from the true ranking of candidates since they may be biased for or against specific candidates. The biased jury problem was initiated by Amor\'{o}s et al. \cite{amoros2002scholarship} and Amor\'{o}s \cite{amoros2009eliciting} in a series of very interesting papers.\footnote{See also  Amor\'{o}s \cite{amoros2011natural},  Adachi \cite{adachi2014natural}, Amor\'{o}s \cite{amoros2013picking},  Amor\'{o}s\cite{amoros2016subgame}, Amor\'{o}s \cite{amoros2025eliciting}, Yadav \cite{yadav2016selecting}.}   A significant difference between the biased jury and peer selection literature is that in the former, the selectors (the jury) are not themselves candidates. Another  difference   between our paper and the papers mentioned here is that they focus on deterministic mechanisms and$\backslash$or different solution concepts such as Nash equilibrium or subgame perfection.

   \section{The Model}

There is a set  $V = \{1,\ldots,n\}$ of agents.  Nature selects an agent to be   needy with probability $q >0$.  Let 
 $N$   denote the realized set of \textit{needy} agents. 
 
Any pair  $i, j \in V\times V$  are either \textit{friends}, \textit{enemies} or \textit{impartials}. Of course,  these are symmetric relationships. 
 For any $i \in V$, 
$ F_i$, $E_i$, $I_i$ are the set of friends, enemies and impartials of $i$. So, $\{F_i, E_i, I_i\} $ form a \textit{partition} of $V\setminus \{i\}$. 

We assume that  every agent in $V$ knows the set $N$ of needy agents as well as her own set of friends, enemies and impartials. We are agnostic about whether $i$ knows anything else about the  networks of friends, enemies and impartials.    We can then represent the type of any agent as 
 a triple $\theta_i= (N, F_i, E_i)$. Notice that we can ignore $I_i$ since $I_i\equiv V \setminus (\{i\} \cup F_i\cup E_i)$. 
 Let $\Theta_i$ represent the set of all possible types of agent $i$. 

The planner does not know the set $N$. We consider two alternative informational scenarios. In one scenario, the planner knows either the sets $I_i$ or   $F_i$ or $E_i$  of each agent.  In the other informational scenario, the planner knows nothing about  these sets.  In either case, the planner wants to use a mechanism that will induce agents to report the set $N$ truthfully.

\textbf{Mechanism}

We are interested in direct mechanisms where each agent reports her type. So, a typical  message profile, to be denoted  by $\bm{m}$ will be an element of  $\bm{\Theta} = \prod_{i \in V} \Theta_i$. 
A (direct) mechanism is defined by an outcome function $g : \bm{\Theta} \rightarrow [0,1]^V$, with $g_i(\bm{m} ) $ representing the probability of  agent $i$ being chosen.   The mechanism is \textit{valid} if for all  message profiles $\bm{m}$,
\[ \sum_{i \in V} g_i(\bm{m}) \le 1.\]
Notice that this keeps open the possibility that the planner may not choose any agent at some message profiles. 

Throughout, the rest of the paper we restrict attention to valid mechanisms. 

\textbf{Preferences}

Let $p=(p_1, \ldots, p_n)$ and $p'=(p'_1, \ldots, p'_n)$ be two vectors where $p_i, p_j$, $p'_i,p'_j$ etc. are probabilities  with which $i$ and $j$ are chosen corresponding to  different message profiles.

Agents have lexicographic preferences over   such probability vectors. First,  $i$ checks her own probability of being chosen.  She strictly prefers $p$ to $p'$ if $p_i >p'_i$. 

Second,  she prefers  her friends being chosen and dislikes her enemies being chosen. She assigns weights  $w_f >0, w_e>0$ and  ranks vectors $p$ and $p'$ so that 
\begin{equation*} 
p\succ_i p' \mbox{ if } \begin{cases} p_i >p'_i \\
p_i=p'_i \mbox{ and  } \sum_{j \in F_i} w_f ( p_j -p'_j) - \sum_{j \in E_i} w_e ( p_j -p'_j) >0\end{cases}\end{equation*}

It is clear that in our setting, preferences do not depend upon the set $N$. So, for instance, each individual $i$  is indifferent between whether friend $j$ or friend $k$ is selected, regardless of the neediness status of $j$ and $k$.

A mechanism $g$ is \emph{dominant strategy incentive compatible} (DSIC) if and only if, for all $i \in V$, all $\bm{m} \in \bm{\Theta}$ and all $m'_i \in \Theta_i$, 
\begin{equation*}
p \succsim_i p', \mbox{ where } p= g(\bm{m}), p'= g(m'_i,\bm{m}_{-i}).
\end{equation*}
That is, revealing the true state must be  a dominant strategy for $i$ at all states.

The planner also wants to be able to pick a needy agent whenever the set of needy agents is non-empty. This gives rise to the next definition. 

 A mechanism $g$ is \emph{efficient} if and only if, for any state of the world, $\theta$ with truthful message profile $\bm{m} \in \bm{\Theta}$ , if $N \ne \varnothing$, then $\sum_{i \in N } g_i(\bm{m})=1$.

 \section{Known Relationships Networks}
 
 In this section, we assume that the planner knows either the network of impartials or the network of enemies or the network of friends. In each case, we derive sufficient conditions under which the planner can construct DSIC mechanisms that are also efficient.
 
Since the planner is assumed to know the relevant networks, all the mechanisms considered in this section ignore reports about the sets $I_i, E_i, F_i$, and only use the reports on the neediness status of others.  So, in order to simplify the notation,   we will   identify message $m_i$ sent by an agent $i$ with the set of needy agents reported by $i$. In particular, $j\in m_i$ means that $i$ reports that $j$ is needy and $j \notin m_i$ means that $j$ is not needy according to $i$.

 \medskip
 
 \subsection{ Planner knows the sets $I_i$ of each $i$}
 
 \medskip

The mechanism we construct  exploits the fact that impartials do not have an incentive to lie about the neediness status of their impartial neighbours provided their reports do not affect their own probability of being chosen.

 Given a message profile $\bm{m}$ and any pair of agents $i, j \in V$, the  message $m_j$ is called a \emph{positive vote of $j$ on $i$} if $i \in m_j$, while it is called a \emph{ negative vote}  if $i\notin m_j$.   In addition, given a message profile $\bm{m} $ and a set of agents $X \subseteq V$, $\bm{m}_{X} = \{m_j\}_{j \in X}$ is the message profile $\bm{m}$ restricted to agents in $X$. In the case of $X = V$, we simply write $\bm{m}$ instead of $\bm{m}_V$ and in the case of $X = V\setminus \{i\}$ we write $\bm{m}_{-i}$ instead of $\bm{m}_{V\setminus \{i\}}$.

Given a message profile $\bm{m}$, let $\pv(\bm{m}_X)$ be the set of agents $j\in V$ receiving positive votes from all their impartial agents in $X$, i.e. from all agents in $X\cap I_j$.

We construct the following mechanism $g^1$. 
The probability of agent $i\in V$ being selected under message profile $\bm{m}$ is
\begin{equation*}
g^1_i(\bm{m}) = \begin{cases}
                \frac{1}{|\pv(\bm{m}_{I_i})|}, & \textrm{if $i \in \pv(\bm{m}_{I_i})$,}\\
                0, & \textrm{otherwise.}
                \end{cases}
\end{equation*}

So, for each agent $i \in V$, the mechanism first checks whether the agent receives positive votes from all her impartial  agents.  If $i$ passes this test, then  the mechanism  assigns $i$ the probability that is the reciprocal of  the number of all the agents, $j \in V$ who received positive votes from the agents that are impartial  to both $i$ and $j$,  i.e. all the agents in $I_i \cap I_j$.  If $i$ receives even one negative vote from some $j \in I_i$,  then $g^1$ assigns zero probability to $i$. 

The following condition will play an important role.

\medskip
\textbf{ The Intersection  Condition I}:  For all $i,j \in V$, $I_i \cap I_j \neq \varnothing$. 

\medskip

The next result describes  the properties of $g^1$. 
\begin{theorem}
\label{pr:g1}
Suppose the planner knows the sets $I_i$ for all $i \in V$. Then,
 $g^1$ is valid and DSIC. Moreover, if the Intersection Condition I is satisfied,  then $g^1$ is efficient.
\end{theorem}

\begin{proof}
To  check validity, let
\begin{equation*}
q_i(\bm{m}) = \begin{cases}
              \frac{1}{|\pvi(\bm{m})|}, & \textrm{if $i \in \pvi(\bm{m})$},\\
              0, & \textrm{otherwise}.
              \end{cases}
\end{equation*}
Given a message profile $\bm{m}$, $\pvi(\bm{m})$ is the set of agents who receive positive votes from all their respective impartial agents. The quantity $q_i(\bm{m})$ equals $0$ if $i$ receives a negative vote from at least one of her impartial agents and the reciprocal of the number of agents who receive positive vote from all their respective impartial agents, otherwise.
To show validity, we will show that for any message profile, $\bm{m}$, and any agent $i \in V$, $g^1_i(\bm{m}) \leq q_i(\bm{m})$.

If $g^1_i(\bm{m}) = 0$ then the inequality is obviously satisfied. Suppose that $g^1_i(\bm{m}) > 0$. Then $i \in \pvi(\bm{m}_{I_i})$, that is $i$ receives a positive vote from all agents in $I_i$. Hence $i \in \pvi(\bm{m})$. Moreover, for any $j \in \pvi(\bm{m})$, $j$ receives a positive vote from all agents in $I_j$. Therefore $j$ receives positive vote from all agents in $I_i \cap I_j$ and, consequently,
$j \in \pvi(\bm{m}_{I_i})$. 

Thus $\pvi(\bm{m}) \subseteq \pvi(\bm{m}_{I_i})$ and $1/|\pvi(\bm{m})| \geq 1/|\pvi(\bm{m}_{I_j})|$.
Hence $g^1_i(\bm{m}) \leq q_i(\bm{m})$.

Using this,
\begin{equation*}
\sum_{i \in V} g^1_i(\bm{m}) \leq \sum_{i \in V} q_i(\bm{m}) = 1,
\end{equation*}
which implies validity of $g^1$.

DSIC follows because, for any agent $j$, $g^1_j(\bm{m})$ depends on $\bm{m}_{I_j}$ only. Since $j\notin I_j$, $g^1_j(\bm{m})$ is independent of the report of $j$. Moreover, for any agent $i$ that is not impartial to $j$, $\bm{m}_{I_i}$ is independent of the report of $j$. Hence, for any $i \notin I_j$, $g^1_i(\bm{m})$ is independent of the report of $j$.  Hence, truthful reporting is a weakly dominant strategy for $j$. 

For efficiency, suppose that $\bm{m}$ is a truthful message profile. The Intersection Condition I ensures that for all  $i \in V$, $I_i\neq \varnothing$.  So if $i$ is not needy then at least one agent reports this under $\bm{m}_{I_i}$ and so $i \notin \pvi(\bm{m})_{I_i}$. On the other hand, if $i$ is needy then reports in $\bm{m}_{I_i}$ all contain a positive vote on $i$ and so $i \in \pvi(\bm{m}_{I_i})$. Moreover,  any needy agent $j\ne i$  receives a positive vote from all agents in $I_i\cap I_j$. Similarly, if $j $ is not needy,  then she receives a negative vote from some $ k \in  I_i \cap I_j $.  

Hence, for any $\{i,j\} \subseteq V$, $j \in \pvi(\bm{m}_{I_i})$ if and only if $j$ is needy. Thus it follows that, for any agent $i \in V$, $g^1_i(\bm{m}) = 0$ if $i$ is not needy and $g^1_i = 1/z$, where $z$ is the number of all needy agents, if $i$ is needy.
\end{proof}
\begin{remark}
To see the role played by the Intersection Condition I, consider the following example. Let $V=\{1,2,3\}$, $I_1 = \{2\}, I_2=\{1\}$, and $N=\{1\}$. The Intersection Condition is not satisfied because for instance $I_1 \cap I_2$ and $I_1\cap I_3 $  are  empty.  At  the truthful report $\bm{m}$, $g_2(\bm{m}) = 0$ since $1 \in I_2$ and gives $2$ a negative vote. Consider $3$ who has no impartials. So, trivially $3$ gets a positive vote from all her impartials. For similar reasons, for $i=1,2$, $i$  also gets positive votes from all agents in $I_i\cap I_3$ for $i=1,2$.  So, $g^1_3(\bm{m})=1/3$. It also follows that $g^1_1(\bm{m}) =1/3$.  Of course, efficiency required $g^1_1(\bm{m}) =1$.

\end{remark}

\subsection{The planner knows the sets $E_i$ for all agents $i \in V$}

\medskip

Suppose now that the planner knows the enemy network, that is the sets $E_i$ for all agents $i$.  We will construct another DSIC mechanism that will use reports of enemies, and that will,  under an Intersection Condition which is similar to the earlier condition, be also efficient.  At this stage, it is worth emphasizing the similarities and differences between a DSIC mechanism  based on messages from impartials and one based on reports from enemies.  A common implication of DSIC mechanisms is that no agent $i$ can unilaterally change her own probability of being selected.   As exhibited by $g^1$, this  is essentially the only condition that is imposed by DSIC if the mechanism uses only messages from agents  on the neediness status 
of their impartials.  However, a  DSIC mechanism that uses messages from enemies must satisfy an additional   requirement -- no agent $i$ should be able to gain by declaring an enemy to be not needy if she is actually needy. Since agent $i$ would like to declare that enemy $j$ is not needy, the mechanism must ensure that the \textit{total probability} of agent $i$'s enemies being  selected does not depend upon agent $i$'s report. 

The mechanism we construct possesses this feature because it uses a  distinguished  agent as a ``sink''.   In particular, there must be a special agent $k$ such that the enemies of $k$, the set $E_k$, can act as a set of \textit{selectors}.  The set of selectors must have the property that for any pair of agents $i, j \in V\setminus \{k\}$, there is an agent $l  \in E_k$ who is a common enemy of $i, j$. Since both $i$ and $j$ are enemies of $l$, agent $l$ has no incentive to discriminate between the pair. Moreover, $l$ also does  not gain by declaring that $i$ or $j$ is not a needy agent if they actually are needy precisely because $k$ acts as a sink and $l$ is also an enemy of $k$. Notice that since this new intersection condition (that $l \in E_i\cap E_j$) must also apply to pairs of agents in $E_k$, $E_k$ must contain at least three agents.

We describe formal  details below.  

Given a message profile $\bm{m}$, let $\pve(\bm{m}_X)$ be the set of agents, $j\in V$, receiving positive votes from all their \emph{enemies} in $X$, i.e. from all agents in $X\cap E_j$.

For any  distinguished agent $k \in V$,  the  mechanism $g^{2,k}$ is  defined as follows. 

The probability of agent $i\in V\setminus\{k\}$ being selected under message profile $\bm{m}$ is
\begin{equation*}
g^{2,k}_i(\bm{m}) = \begin{cases}

                \frac{1}{|\pve(\bm{m}_{E_i \cap E_k})\setminus \{k\}|}, & \textrm{if $i \in \pve(\bm{m}_{E_k})$,}\\
                0, & \textrm{otherwise.}
                \end{cases}
\end{equation*}
The probability of agent $k$ being selected is
\begin{equation*}
g^{2,k}_k(\bm{m}) = 1 - \sum_{i \in V\setminus \{k\}} g^{2,k}_i(\bm{m}).
\end{equation*}
So, for each agent $i \in V\setminus\{k\}$ the mechanism checks whether the agent receives positive votes from all her enemies who are also enemies of $k$, i.e. from all agents in $E_i\cap E_k$, and if so, it assigns $i$ the probability equal to $1$ over the number of all the agents, $j \in V\setminus \{k\}$ who received positive votes from all their enemies who are also enemies of $i$ and $k$, i.e. all the agents in $E_j \cap E_i \cap E_k$. If $i$ receives a negative vote from at least one of her enemies who is also an enemy of $k$, then the probability assigned to $i$ is $0$. Agent $k$ is selected whenever no other agent is selected, i.e. $k$ is selected with probability equal to $1$ minus the probability of selecting one of the agents in $V\setminus \{k\}$.

\medskip
\textbf{Intersection Condition E(k)}: For all $i, j \in V\setminus\{k\}$, $E_i \cap E_j \cap E_k \neq \varnothing$. 

\medskip

The following result  states validity, efficiency and DSIC of $g^{2,k}$ when the Intersection Condition E(k)  is satisfied.

\begin{theorem}
\label{pr:g2}
For any $k \in V$ mechanism $g^{2,k}$ is valid and DSIC. 
Moreover, if the Intersection Condition E(k)  is satisfied, then $g^{2,k}$ is efficient.
\end{theorem}

\begin{proof}
The proof of validity is very similar to the proof of validity of $g^1$. 

For an agent $i \in V\setminus \{k\}$, let
\begin{equation*}
q_i(\bm{m}_{E_k}) = \begin{cases}
              \frac{1}{|\pve(\bm{m}_{E_k})\setminus \{k\}|}, & \textrm{if $i \in \pve(\bm{m}_{E_k})$},\\
              0, & \textrm{otherwise}.
              \end{cases}
\end{equation*}
Given a message profile $\bm{m}$, $\pve(\bm{m}_{E_k})$ is the set of agents who receive positive votes from all their enemies who are also enemies of $k$. The quantity $q_i(\bm{m}_{E_k})$ gets value $0$ if $i$ receives a negative vote from at least one of her enemies who is an enemy of $k$ and the reciprocal of  the number of agents in $V\setminus \{k\}$ who receive positive vote from all their respective enemies other than $k$, otherwise.

To show validity, we will show that for any message profile, $\bm{m}$, and any agent $i \in V\setminus \{k\}$, $g^{2,k}_i(\bm{m}) \leq q_i(\bm{m}_{E_k})$.
If $g^{2,k}_i(\bm{m}) = 0$ then the inequality is satisfied. Suppose that $g^{2,k}_i(\bm{m}) > 0$. Then $i \in \pve(\bm{m}_{E_k})$. Moreover, for any $j \in \pve(\bm{m}_{E_k})$, $j$ receives a positive vote from all agents in $E_j\cap E_k$ and, consequently, $j$ receives a positive vote from all agents in $E_i\cap E_j\cap E_k$. Hence $j \in \pve(\bm{m}_{E_i \cap E_k})$. This shows that $\pve(\bm{m}_{E_k}) \subseteq \pve(\bm{m}_{E_i\cap E_k})$ and $1/|\pve(\bm{m}_{E_k})| \geq 1/|\pve(\bm{m}_{E_i\cap E_k})|$. Hence $g^{2,k}_i(\bm{m}) \leq q_i(\bm{m}_{E_k})$.
Since
\begin{equation*}
\sum_{i \in V\setminus\{k\}} g^{2,k}_i(\bm{m}) \leq \sum_{i \in V\setminus\{k\}} q_i(\bm{m}_{E_k})= 1
\end{equation*}
and
\begin{equation*}
g^{2,k}_k(\bm{m}) = 1 - \sum_{i \in V\setminus\{k\}} g^{2,k}_i(\bm{m}) \geq 0,
\end{equation*}
we get  validity of $g^{2,k}$.

For DSIC notice first that, for any agent $j \in V$, $g^{2,k}_j$ depends on messages of agents in $E_j \cap E_k$ only.  So, $j$ cannot affect her own probability of being chosen.

Hence it is enough to show that truthful report is a dominant strategy for any agent $i \in E_k$. 
Take any agent $i \in E_k$. 
Fix a message profile $\bm{m}$. Since $i \notin E_i$ and $g^{2,k}_i$ depends on  $\bm{m}_{E_i \cap E_k}$ only,  $g^{2,k}_i$ is independent of $m_i$. 
Next, for any agent $j \in F_i$, $j \notin E_i$, and so $g^{2,k}_j$ is independent of $m_i$. Therefore,
\begin{equation}\label{eq1}
\sum_{j \in F_i} g^{2,k}_j(m'_i,\bm{m}_{-i}) = \sum_{j \in F_i} g^{2,k}_j(\bm{m}).
\end{equation}
Moreover,  the definition of $g^{2,k}$ makes it clear that 
\[\sum_{j\in V} g^{2,k}_j (\bm{m}) \equiv 1, \mbox{ for all } \bm{m} \in \bm{\Theta} \] 
Hence, using equation \eqref{eq1},
\[ \sum_{j \in E_i} g_j^{2,k} (m'_i,\bm{m}_{-i}) = \sum_{j \in E_i} g_j^{2,k}(\bm{m})  \forall m'_i.\]

Hence, no agent $i$ gains by deviating from truthful reporting.

Suppose now that the Intersection Condition E(k) is satisfied.
To check  efficiency, suppose that $\bm{m}$ is a truthful message profile.  From the Intersection Condition E,  it follows that if $i$ is not needy then at least one agent reports that under $\bm{m}_{E_i\cap E_k}$ and so $i \notin \pve(\bm{m}_{E_k})$. On the other hand, if $i$ is needy then reports in $\bm{m}_{E_i \cap E_k}$ all contain a positive vote on $i$, and so $i \in \pve(\bm{m}_{E_i\cap E_k})$. Moreover, for any needy agent $j \in V\setminus \{k\}$, $j$ receives a positive vote from all agents in $E_j\cap E_i\cap E_k$.  Conversely,   any non-needy agent $j\in V\setminus \{k\}$  receives a negative vote from at least one agent in $ E_j \cap E_i \cap E_k$. Hence, for any $\{i,j\} \subseteq V\setminus \{k\}$, $j \in \pve(\bm{m}_{E_i \cap E_k})$ if and only if $j$ is needy. Thus if follows that, for any agent $i \in V\setminus \{k\}$, $g^{2,k}_i(\bm{m}) = 0$ if $i$ is not needy and $g^{2,k}_i(\bm{m}) = 1/z$, where $z$ is the number of all needy agents in $V\setminus \{k\}$, if $i$ is needy. If there is no needy agent in $V\setminus \{k\}$ then $\sum_{i \in V\setminus \{k\}} g^{2,k}_i(\bm{m}) = 0$ and, consequently, 
$g^{2,k}_k(\bm{m}) = 1$.  In this case, either $k$ is needy and the support of $g^{2,k}(\bm{m})$ is contained in $N$. Or $k$ is not needy, implying that $N=\varnothing$ and so $g^{2,k}$ is efficient. 
\end{proof}

\medskip

\subsection{The planner knows the sets $F_i$ for all agents $i \in V$}

\medskip

Suppose instead that the planner knows the sets $F_i$ of each agent.  It is not difficult to see that a DSIC mechanism using messages only of friends analogous to $g^{2,k}$ can be defined with a pre-specified agent $k$ acting as a sink.

Given a message profile $\bm{m}$, let $\pvf(\bm{m}_X)$ be the set of agents, $j\in V$, receiving positive votes from all their \emph{friends} in $X$, i.e. from all agents in $X\cap F_j$. 
Let  mechanism $g^{3,k}$  be defined as follows. 
The probability of agent $i\in V\setminus\{k\}$ being selected under message profile $\bm{m}$ is
\begin{equation*}
g^{3,k}_i(\bm{m}) = \begin{cases}
                \frac{1}{|\pvf(\bm{m}_{F_i \cap F_j})\setminus \{k\}|}, & \textrm{if $i \in \pvf(\bm{m}_{F_k})$,}\\
                0, & \textrm{otherwise.}
                \end{cases}
\end{equation*}
The probability of agent $k$ being selected is
\begin{equation*}
g^{3,k}_k(\bm{m}) = 1 - \sum_{i \in V\setminus \{k\}} g^{3,k}_i(\bm{m}).
\end{equation*}

\medskip

\textbf{Intersection Condition F(k)}: For all $i, j \in V\setminus\{k\}$, $F_i \cap F_j \cap F_k \neq \varnothing$.

\medskip

\begin{theorem}
\label{pr:g3}
For any $k \in V$ mechanism $g^{3,k}$ is valid and DSIC. 
Moreover, if the Intersection Condition F(k) is satisfied, then $g^{3,k}$ is efficient.
\end{theorem}

The proof is omitted since it is virtually identical to that of Theorem~\ref{pr:g2}.

\medskip

 \section{Unknown Relationships Networks}
 
 In the last section, we showed that the planner can construct efficient and DSIC mechanisms for classes of networks satisfying appropriate sufficient conditions provided the planner knows the networks.  Notice that the construction of these mechanisms exploits the finer details of the networks.  For instance, when the network satisfies the Intersection Condition I, the mechanism used in the proof of Theorem~\ref{pr:g1} depends only on the network of impartials. This raises the question of whether the planner can construct efficient DSIC mechanisms  when she does not know  the networks of impartials, enemies and friends.   Notice that in this case, the planner has to construct \textit{ one } mechanism that is DSIC and efficient for \mbox{ all } networks. In the theorem below, we show that this cannot  be done.  
 Since we want to focus on the role played by incomplete information about the network,  the proof below will only use networks that satisfy the Intersection Condition I. This will highlight the fact that the impossibility result is precipitated by the lack of information about the bilateral relationships.

\begin{theorem}
\label{th4}  
Let $|V| =n \ge 4$. If the planner does not know the sets of impartials, enemies and friends, then there is no mechanism that is both DSIC and efficient.
\end{theorem}
\medskip
\begin{proof}

Let $V=\{1, 2,\ldots, n\}$ with  $n \ge 4$.

 Consider two states of the world, $\bm{\theta}^1$ and $\bm{\theta}^n$. In $\bm{\theta}^ 1$,   only agent $1$ is needy, while in $\bm{\theta}^n$, only agent $n$ is needy. The bilateral relationships in the two states are as follows.
\begin{enumerate}[(i)]
\item In state of the world $\bm{\theta}^1$,   all pairs of agents $i, j \in S\equiv  \{1,2, \ldots, n-2\}$ are enemies while all  remaining pairs of agents are impartial.
\item In state of the world  $\bm{\theta}^n$,  agents $n-1$ and $n$ are enemies,  and all the remaining pairs of agents are impartial. \footnote{
Notice that  the networks in both states  satisfy the intersection condition for impartials when $n\ge 4$.} 
\end{enumerate}

Figure~\ref{fig:noeffdsic} presents the two networks for the case of $n = 4$ agents.  The figure on the left shows the state $\bm{\theta}^1$, with the arc between $1$ and $2$ representing that they are enemies and $1$ is the needy agent, while the figure on the right represents state $\bm{\theta}^n$.

\begin{figure}[h]
\begin{center}
\includegraphics[scale=.75]{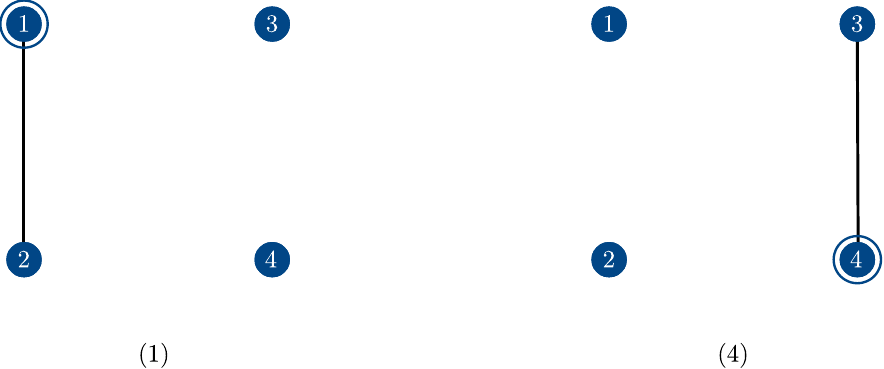}
\end{center}
\caption{Two states of the world with $n = 4$ agents. The encircled nodes correspond to the needy agents.}
\label{fig:noeffdsic}
\end{figure}

Let $\bm{m}^i$ be the truthful message profile at state of the world  $\bm{\theta}^i$, $i \in \{1,n\}$. Then, efficiency  requires that 
\begin{eqnarray}
g_1(\bm{m}^1) &=& 1 \label{th4:eq1}\\
g_n(\bm{m}^n) &=& 1 \label{th4:eq2}.
\end{eqnarray}

From ~\eqref{th4:eq1} and DSIC, 
\begin{equation}\label{th4:eq3} g_1(m^n_1,\bm{m}_{-1}^1) = 1\end{equation}
 For suppose  $g_1(m^n_1,\bm{m}_{-1}^1) < 1$.  Then, agent $1$ would strictly prefer to report $m^1_1$ instead of $m^n_1$  when  the other agents report $\bm{m}_{-1}^1$. This would contradict the supposition that $m^n_1$ is a dominant strategy for $1$ when the true state of the world is $\bm{\theta}^n$. 
 
Next, consider agent $2$. Since $m^1_2$ is a dominant strategy in state of the world $\bm{\theta}^1$, 
a unilateral deviation to $m^n_2$ cannot change the sum of the probabilities with which  enemies of $2$ are selected. So,  ~\eqref{th4:eq3} implies that
\[\sum_{j \in S\setminus\{2\}}g_j(m^n_1, m^n_2, m^1_{-\{1,2\}}) =1\]

 Similarly, since $g^1_3$ is  a dominant strategy in state $\bm{\theta}^1$, agent 3 cannot change her own probability or that of her enemies through a unilateral deviation. Hence, 
 \[\sum_{j \in S}g_j(m^n_1,m^n_2,m^n_3,\bm{m}_{-\{1,2,3\}}^1) = 1 \]
By repeated use of the same argument. 
\begin{equation}
\label{th4:eq4}
\sum_{j \in S}g_j(m^n_1,m^n_2,\ldots, m^n_{n-2}, m^1_{n-1},m^1_n) = 1.
\end{equation}

But, notice that we could have proceeded in the \lq\lq reverse direction" starting from equation  ~\eqref{th4:eq2}, first considering the deviation of agent $n$ from $m^n_n$ to $m^1_n$ and then that of agent $2$ from $m^n_{n-1} $ to $m^1_{n-1}$. By analogous reasons, we would get
\[\sum_{j \notin S}g_j(m^n_1,m^n_2,\ldots, m^n_{n-2}, m^1_{n-1},m^1_n) = 1.\]
This contradicts ~\eqref{th4:eq4}.

This completes the proof of the theorem.
\end{proof}

\medskip

In view of Theorem~\ref{th4}, the search must be for \lq\lq second-best" mechanisms.  Of course, the constant mechanism where each agent is picked with equal probability is DSIC. The constant mechanism picks a needy agent with probability $q$. 

Can one improve on the constant mechanism?  We  take
 a step in this direction by constructing two  DSIC mechanisms  and comparing their relative efficiency.   The mechanisms we construct are adaptations of the two canonical DSIC mechanisms -- the \textit{random dictatorship} and \textit{duples} mechanisms -- in the literature on probabilistic social choice starting from Gibbard~\cite{gibbard1977manipulation}.

  We point out the differences between our framework and that of Gibbard~\cite{gibbard1977manipulation} before we formally describe the two mechanisms.  

In our framework, voters and alternatives coincide, and it is commonly known that  every agent $i$'s most preferred outcome is $i$ herself. Hence, the unmodified random dictatorship where  each agent is picked at random and that  agent's most preferred outcome is chosen is nothing but the constant mechanism itself.   Moreover, there is some correlation in individual preference since the bilateral relationships of friendship, impartiality and enmity are symmetric.\footnote{In principle, this allows the designer to detect lies in case of discordant reports. For instance, if $i$ reports that $j$ is a friend but $j$ does not confirm this, then one of $i$ or $j$ is lying. This does not necessarily help the designer to construct DSIC mechanisms since there could be \textit{coordinated} lies.} Perhaps the most important difference between the two frameworks is that preference orderings are not strict in our framework unlike  the Gibbardian framework. Finally, our entirely different  notion of efficiency  is based on the external characteristic of neediness and is quite independent of individual preferences. 

We now define the two DSIC mechanisms.  Both are direct mechanisms and so each agent reports a type $m_i= ( F_i, I_i, E_i, N)$ or simply one of the two sets of relationships since the latter follows as a residue.   
 Given a message profile $\bm{m}$ and an agent $j$, let $E_j(\bm{m})$, $F_j(\bm{m})$, and $I_j(\bm{m})$ be sets of enemies, friends, and impartial  agents of agent $j$,  and $N_j(\bm{m})$ be set of needy agents reported by $j$. 
 
Formally,  let  $g^{\mathbf{R}}$ be the random dictatorship mechanism defined as follows

\begin{equation*}
g^{\mathrm{R}}_{i}(\bm{m}) = \frac{1}{n} \sum_{j \in V\setminus \{i\}} g^{\mathrm{R}}_{i \mid j} (\bm{m}), 
\end{equation*}
where \begin{equation*}
g^{\mathrm{R}}_{i\mid j}(\bm{m}) = \begin{cases}
                \frac{1}{|F_j(\bm{m}) \cap {N}_j(\bm{m})|}, & \textrm{if $i \in F_j(\bm{m}) \cap {N}_j(\bm{m})$,}\\
                \frac{1}{|F_j(\bm{m})|}, & \textrm{if $i \in F_j(\bm{m})$ and $F_j(\bm{m}) \cap {N}_j(\bm{m}) = \varnothing$,}\\
                \frac{1}{|I_j(\bm{m}) \cap {N}_j(\bm{m})|}, & \textrm{if $i \in I_j(\bm{m}) \cap {N}_j(\bm{m})$ and $F_j(\bm{m}) = \varnothing$,}\\
                \frac{1}{|I_j(\bm{m})|}, & \textrm{if $i \in I_j(\bm{m})$ and $F_j(\bm{m}) \cup (I_j(\bm{m}) \cap {N}_j(\bm{m})) = \varnothing$,}\\
                \frac{1}{|E_j(\bm{m}) \cap {N}_j(\bm{m})|}, & \textrm{if $i \in E_j(\bm{m}) \cap {N}_j(\bm{m})$ and $F_j(\bm{m}) \cup I_j(\bm{m}) = \varnothing$,}\\
                \frac{1}{|E_j(\bm{m})|}, & \textrm{if $i \in E_j(\bm{m})$ and $F_j(\bm{m}) \cup I_j(\bm{m}) \cup (E_j(\bm{m}) \cap {N}_j(\bm{m})) = \varnothing$,}\\
                0, & \textrm{otherwise.}
                \end{cases}
\end{equation*}

As is standard, the  random mechanism picks  each agent with probability $1/n$,\footnote{Of course,  agents can be picked according to any  fixed probability distribution  that is  independent of the message profile. The equiprobability specification is more equitable.}  and allowing that  agent to specify her set of \textit{second best} outcomes.    So, if $i$ has a needy friend, then she gets to pick such  friends with equal probability. If no friend is needy, then she gets to choose her friends with equal probability even if they are not needy. If $i$ has no friends, then she chooses   needy impartial agents if such  agents exist. If not, then she chooses   impartial agents with equal probability. Finally, if $i$ has only enemies, then she chooses the  needy enemies, and so on. 
Notice that the sequence in which $i$ is allowed to choose matters.  For instance, if she is asked to choose initially from her set of enemies, then she will declare a friend to be an enemy and choose her.

Next, we define the duples mechanism.  For any $m_i$, define  a hierarchy of sets $H(m_i) =\{ S^1(m_i), S^2(m_i), S^3(m_i), S^4(m_i), S^5(m_i), S^6(m_i)\}$ such that 
\begin{itemize}
\item []$S^1(m_i) = F_i \cap N$
\item[] $S^2(m_i)  = F_i \setminus N$
\item[] $S^3(m_i)  = I_i \cap  N$
\item[] $S^4(m_i) = I_i \setminus N$
\item []$S^5(m_i) = E_i \cap N$
\item [] $S^6(m_i) = E_i \setminus N.$
\end{itemize} 
For any pair of agents, $j,k \in V\setminus\{i\}$.
\[ \mbox{ Agent $i$ votes for $j$ against $k$ if } j\in S^a(m_i), k \in S^b(m_i), a<b\]
Agent $i$ abstains from voting over the pair $\{j,k\}$  if $j$ and $k$ both belong to $S^a(m_i) $ for some $a$.

Let $x_{jk}(\bm{m}) $ be the number of votes cast for $j$, and $x_{kj}(\bm{m}) $ be the number of votes cast for $k$ in the vote over the pair$\{j,k\}$ at report profile $\bm{m}$.  Then,

\[
g_j^D({\{j,k\}, \bm{m}) =\begin{cases} 1, & \mbox{ if } x_{jk}(\bm{m}) > x_{kj}(\bm{m}),\\
1/2, &\mbox{ if } x_{jk}(\bm{m}) = x_{kj}(\bm{m}),\\
0, & \mbox{ if } x_{jk}(\bm{m}) < x_{kj}(\bm{m}),
\end{cases}}
\]
And,
\[g_k^D(\{j,k\},\bm{m}) =1- g_j^D(\{j,k\}, (\bm{m}).\]
Then, the duples mechanism, $g^D$ is 
\[g_i^D(\bm{m}) = \frac{2}{n(n-1)} \sum_{j \in V\setminus\{i\}} g_i^D(\{i,j\}, \bm{m}). \]

The following proposition is stated without proof.
\begin{proposition}
Both $g^R $ and $g^D$ are valid and DSIC mechanisms.
\end{proposition}

We now evaluate the relative efficiency of the random dictator and duple mechanisms.

Fix the networks of friends, enemies and impartials and denote the overall network by $G$.
Consider any agent $j$  in the network with $f >0$ friends. Suppose  $j$ is the random dictator. Then,  the probability that none of her friends is needy is $(1-q)^f$. So, when $f>0$, agent $j$ selects a needy individual with probability $1- (1-q)^f$. When agent $j$ has no friends, she selects a needy agent with probability $1-(1-q)^i$ when $i>0$ is the number of her impartials.  Finally, if agent $j$ has only enemies, she selects a needy agent with probability $1-(1-q)^{n-1}$.  

So, letting $m^G_{f,i}$  be the number of nodes with $f$ friends and $i$ impartials in the given network $G$,   the probability of choosing a needy agent in $G$ by the random dictator mechanism is:
\begin{align}\label{rd}
P_{RD}(G) = & \frac{1}{n} \Bigg( m^{G}_{0,0} \left(1 - (1-q)^{n-1}\right) + {} \nonumber\\
            & \sum_{i = 1}^{n-1} m^{G}_{0,i} \left(1 - (1 - q)^{i}\right) + {} \nonumber\\
            & \sum_{f = 1}^{n-1} \sum_{i = 0}^{n - f - 1} m^{G}_{f,i} \left(1 - (1-q)^f\right)\Bigg).
\end{align}
Hence the probability of choosing a needy agent by the random dictator mechanism, on any network $G$, is greater or equal to $q$.  The following proposition describes the obvious implications of Equation~\eqref{rd}.  

\begin{proposition}\label{prop:rd}
The random dictatorship mechanism 
\begin{itemize}
\item[(i)] coincides with the constant mechanism  iff every agent either has exactly one friend or no friend and exactly one impartial.
\item[(ii)] in all other cases, the probability of a needy agent being selected is equal to $1 - (1-q)^{m}$ where $m \geq 2$.
\item[(iii)]  the highest probability of choosing a needy agent is
\begin{equation*}
P_{RD}(G) = 1-(1-q)^{n-1}
\end{equation*}
and it is attained on a graph where $\max(|F_i|, |I_i|, |E_i|) =n-1$ for all agents.
\end{itemize}
\end{proposition}
We now consider the duples mechanism. Since the duples mechanism is the average of mechanisms each of whom has a  range of two alternatives,   the \textit{maximum} probability with which  the duples mechanism chooses a needy agent is $1-(1-q)^2$.  So, on networks where all agents have two or more friends, or no friends but two or more impartials, or where the enemy network is a complete graph on $V$,  the random dictatorship will be relatively more efficient than the duples mechanism.

It is also interesting to compare the two mechanisms in terms of a worst-case scenario. From Proposition~\ref{prop:rd}, we know that the lowest probability with which  a needy agent is picked by the random dictatorship  equals $q$. We now show that 
if there are only two types of links between the agents (e.g. only friends and impartial agents, or only enemies and impartial agents), then  the lowest probability with which the duples mechanism chooses a needy agent is strictly greater than $q$. 

\begin{proposition}
\label{prop:duples}
On any network $G$ such that either $\bigcup_{i\in V} E_i = \varnothing$, or $\bigcup_{i\in V} F_i = \varnothing$, or $\bigcup_{i\in V} I_i = \varnothing$, the duples mechanism chooses a needy agent with probability strictly greater than $q$.
\end{proposition}

\begin{proof}
Take any network $G$ such that either $\bigcup_{i\in V} E_i = \varnothing$, or $\bigcup_{i\in V} F_i = \varnothing$, or $\bigcup_{i\in V} I_i =\varnothing$.

Notice that in any such network there are only two types of links and for any agent $i$ there is a strict preference between  the two types of neighbours connected by these links (e.g. there are only friends and impartial agents and friends are strictly preferred to impartial agents). Without loss of generality, suppose that $\bigcup_{i\in V} E_i = \varnothing$.

Take any networks $G$ with $n \geq 3$ agents, and  any two distinct agents, $i$ and $j$.  Notice that, under the duples mechanism, the remaining agents would vote as follows:
\begin{itemize}
\item Agents in $F_i \cap F_j$ and $I_i \cap I_j$ vote for a needy agent if one exists.
\item Agents in $F_i\setminus (F_j \cup \{j\})$ vote for $i$, regardless of the neediness of the agents.
\item Agents in $F_j\setminus (F_i \cup \{i\})$ vote for $j$, regardless of the neediness of the agents.
\end{itemize}
These three cases  cover all the agents in $V\setminus \{i,j\}$.
Given the observation above, the probability of choosing a needy agent in graph $G$ from among agents $i$ and $j$ is strictly greater than $q$ if
\begin{equation}
\label{eq:duples:1}
|F_i\setminus (F_j \cup \{j\})| \leq \frac{n-2}{2} \textrm{ and } |F_j\setminus (F_i \cup \{i\})| \leq \frac{n-2}{2},
\end{equation}
with at least one of the inequalities being strict.

To see that this is the case, notice that what matters is which of the agents is chosen when exactly one of them is needy. Suppose that only one of the agents is needy.
If both the inequalities in~\eqref{eq:duples:1} are strict then the mechanism chooses the needy agent. Thus the probability of choosing the needy agent is equal to $1 - (1-q)^2 = q + q(1-q)$ in this case.

If exactly one of the inequalities in~\eqref{eq:duples:1} is strict, say $|F_j\setminus (F_i \cup \{i\})| < \frac{n-2}{2}$, then the mechanism chooses the needy agent if only $i$ is needy.  Suppose only $j$ is needy. Since $ |F_i\setminus (F_j \cup \{j\})| =\frac{n-2}{2}$ and all other agents support $j$, $i$ and $j$ tie and so the needy agent is chosen with probability $1/2$.
 So, the overall probability of choosing a needy agent is $q + q(1-q)/2$ in this case.

If condition~\eqref{eq:duples:1} is not satisfied then, in the case where exactly one of the agents $i$ or $j$ is needy, the mechanism either always chooses the same agent, regardless of the neediness statuses or randomizes uniformly between the two agents. In either case the probability of choosing a needy agent is  equal to $q$.

We complete the proof of the proposition by   showing that in any network $G$ with at least $3$ agents there exists a pair of agents $i$ and $j$ which satisfies condition~\eqref{eq:duples:1}. It is easy to verify that this is true for $n \in \{3,4\}$ agents. Suppose that $n \geq 5$ and assume, to the contrary,
that such a pair of agents does not exist. Then, for any two agents $\{i,j\} \subseteq V$, $i \neq j$,
\begin{equation}
\label{eq:duples:2}
\textrm{either } |F_i\setminus (F_j\cup\{j\})| \geq \frac{n-2}{2} \textrm{ or } |F_j\setminus (F_i\cup\{i\}| \geq \frac{n-2}{2}.
\end{equation}
This implies that there is at most one agent $i$ such that $|F_i| < (n-2)/2$.
This also implies that there is at most one agent $i$ such that $|F_i| > (n-2)/2$. For if there were two agents $i$ and $j$, $i \neq j$, such that $|F_i| > (n-2)/2$ and $|F_j| > (n-2)/2$ then 
\begin{equation*}
|F_i\setminus (F_j\cup\{j\})| \leq n - 1 - |F_j\cup\{j\}| = n - 2 - |F_j| < (n-2)/2
\end{equation*}
and, similarly, $|F_j\setminus F_i| < (n-2)/2$, which contradicts~\eqref{eq:duples:2}.
Hence there are at least $n-2 \geq 3$ agents $i$ such that $|F_i| = (n-2)/2$. This is impossible when $n$ is odd and so we have a contradiction in this case.
Suppose that $n$ is even (in which case $n-2 \geq 4$). Take any two agents, $i$ and $j$, $i\neq j$, such that $|F_i|=|F_j| = (n-2)/2$. This, together with condition~\eqref{eq:duples:2}, implies that $F_i\cap F_j  = \varnothing$. Thus for any two agents $i$ and $j$ such that $|F_i|=|F_j| = (n-2)/2$ it holds that $F_i\cap F_j  = \varnothing$. Since $n -2 \geq 4$ so there exist two agents $k \notin\{i,j\}$ and $l \notin\{i,j\}$ with $|F_k| = |F_l| = (n-2)/2$ and such that either $\{k,l\} \subseteq F_i$ or $\{k,l\} \subseteq F_j$. But then either $i \in F_k\cap F_l$ or $j \in F_k\cap F_l$. On the other hand, by what we established above, $F_k \cap F_l = \varnothing$.  This contradiction completes the proof.
\end{proof}

\section{Structural Balance}

In this section, we restrict attention to the class of networks satisfying 
\textit{structural balance theory}, a concept borrowed from social network analysis which  asserts that social networks tend towards balanced or stable relationships.     Following Heider  \cite{heider1946attitudes},  Cartwright and Harary \cite{cartwright1956structural} formulated the hypothesis within a  graph-theoretic framework.  Using the term \lq\lq friend" to designate a positive sentiment   and the term \lq\lq enemy" to designate a negative relationship, the classic balance model defines a  social  network as balanced if it contains no violations of three assumptions:
 \begin{itemize}
\item If $j \in F_i$ and $k \in F_j$ then $k \in F_i$ (\emph{a friend of a friend is a friend}).
\item If $j \in E_i$ and $k \in E_j$ then $k \in F_i$ (\emph{an enemy of an enemy is a friend}).
\item If $j \in E_i$ and $k \in F_j$ then $k \in E_i$ (\emph{a friend of an enemy is an enemy}).
\end{itemize}
Notice that the third condition, together with symmetry of enmity and friendship relations implies that if $j \in F_i$ and $k \in E_j$ then $k \in E_i$ (\emph{an enemy of a friend is an enemy}).

\subsection{Planner knows the network}

Suppose the planner knows the network. 
We  are then able to fully characterize the networks (within the class of networks satisfying structural balance) which admit an efficient DSIC mechanism.
Before we state the result, we need to introduce the following notions.

An \emph{EF-component} in network $G$ is a maximal set of nodes, $X\subseteq V$, such that for any two nodes $\{i,j\} \subseteq X$, $i \neq j$, there is a sequence of nodes $\{i_0,\ldots,i_m\}$ such that $i_0 = i$, $i_m = j$, and for any $1 \leq r \leq m$, $i_{r-1} \in F_{i_{r}} \cup E_{i_{r}}$. In other words, $j$ is reachable from $i$ in $G$ by a path consisting of the friendship links or enmity links only.

An \emph{F-component} is an \emph{EF-component}  $X$ in which for all pairs of agents $ i,j \in X$, there is a path consisting of the friendship links only.  Of course, every F-component is also an EF-component.

An F-component $X\subseteq V$ of $G$ is an \emph{F-clique} if, for any distinct two nodes $\{i,j\} \subseteq X$,  we have $i \in F_j$. In other words, any two nodes in an F component are friends.

An EF-component $X\subseteq V$ of $G$ is an \emph{EF-clique} if it is either an F-clique or there exists an non-empty set $Y\subsetneq X$,  
such that $Y$ and $X\setminus Y$ are both F-cliques and, for any node $i \in Y$ and $j \in X\setminus Y$, $i \in E_j$, i.e. $i$ and $j$ are enemies.

\begin{lemma} \label{lemma:sb} Let $G$ be a network satisfying Structural Balance. Then,
\begin{itemize}
\item[(i)] Every F-component is an F-clique.
\item[(ii)] Every EF-component is either an F-clique or it consists of two F-cliques, $Y_1$ and $Y_2$, such that for all nodes $i \in Y_1$ and $j\in Y_2$, $i \in E_j$.
\end{itemize}
\end{lemma}

\begin{proof} 
\begin{enumerate}[(i)]
\item Suppose $X$ is an F-component.   Consider three distinct nodes $ i, j, k $ in $X$ such that 
$j \in F_i, k \in F_j$.   Structural balance requires that a friend of a friend is a friend, and so we must have $k \in F_i$. Hence, $X$ is an F-clique.
\item Take an EF-component $X$. If it is an F-component, then it is an F-clique. So, suppose $X$ is not an F-component. Then there are nonempty and mutually disjoint  sets $Y_1,\ldots,Y_m$  such that $X = \bigcup_{r = 1}^{m} Y_r$, and $Y_r$ is an F-component, for all $r \in \{i,\ldots,m\}$. Let $\mathcal{Y}= \{Y_1,\ldots,Y_m\}$ be the set of these F-components.  
\end{enumerate}

Take any pair of sets $Y , Y' \in \mathcal{Y}$. Take any $i \in Y$ and $j \in Y'$. It cannot be that there is a friendship link between $i$ and $j$. For if $i \in F_j$, then since a friend of a friend must be a friend,  and since $Y$ and $Y'$ are F-cliques, \emph{all} pairs of agents in $Y\cup Y'$ must be friends and so $Y \cup Y'$ would be a single F-clique.

Since $X$ is an EF-component, there must be $k \in Y$ and $l \in Y'$ such that $k \in E_l$.  (If no such link existed, then there could not be a path of friendship or enemy links connecting vertices in $Y$ with vertices in $Y'$. ) Using the requirement that a friend of an enemy is an enemy repeatedly,  \emph{for all} pairs of agents $i, j$ with $i \in Y, j\in Y'$, we  must have $i \in E_j$.

Lastly, suppose that $|\mathcal{Y}| \leq 2$. Assume otherwise. Then there exist $Y \in \mathcal{Y}$, $Y' \in \mathcal{Y} \setminus \{Y\}$, and $Y'' \in \mathcal{Y}  \setminus \{Y,Y'\}$ and nodes $i \in Y$, $j \in Y'$, and $k \in Y''$ such that $i \in E_j$ and $j \in E_k$. But then, by structural balance, $k \in E_i$ (an enemy of an enemy is a friend). This contradicts the fact that there are no friendship links between $Y$ and $Y''$.
\end{proof}

We can now state our characterisation result. 

\begin{theorem}  Let $n\ge 4$.  A network $G$ satisfying structural balance admits an efficient and DSIC mechanism if and only if it either has exactly one F-component or it has at least three EF-components.
\end{theorem}
\begin{proof}
Suppose $G$ consists of exactly one F-component. Then, $V$ is an F-clique containing at least 4 nodes. Then, the Intersection Condition F(k) is satisfied for any choice of the ``sink'' $k\in V$.

 If it has at least three EF-components then any two nodes have a common impartial node (in the component they do not belong to) and so the Intersection Condition I is satisfied.

Suppose $G$ consists of (i) one EF-component that is not an F-component, or (ii)  $G$ consists of  exactly two EF-components.

In view of Lemma 1, these give rise to the following possibilities
\begin{itemize}
\item[(a)] $V$ consists of two $F$-cliques $X_1$ and $X_2$, with $E_i= X_2 $ for each $i \in X_1$. This happens if there is one EF-component.
\item[(b)] $V$ consists of F-cliques $X_1, X_2, Y_1, Y_2$\footnote{ Some of these are  possibly empty. However, at least one $X$-set and at least one of the $Y$-sets must be nonempty.} with $E_i =X_2$ for each $i\in X_1$, $E_i = Y_2$ for each $i \in Y_1$, and $I_i= Y_1\cup Y_2$ for each $i \in X_1\cup X_2$.  This happens if there are 2 EF components.
\end{itemize}
The proof in the case of (a) above is very similar to that of Theorem 4 and is avoided.  

 We prove that there is no efficient DSIC mechanism in case (b) above. We assume that all the sets $X_1, X_2, Y_1, Y_2$ are non-empty. The proof can be easily adapted if one of the $X$-sets and one of the $Y$-sets are empty. 
 
  Notice first  
 that none of the Intersection Conditions are satisfied. 
 
Assume to the contrary that $g$ is an efficient DSIC mechanism.

 Suppose that $X_1 = \{ 1,\ldots, l\}$, $X_2= \{l+1, \ldots, s\}$,  $Y_1= \{s+1, \ldots, k\}$,  $Y_2 = \{k+1,\ldots, n\}$. Let there be two states $\bm{\theta}^1$ and $\bm{\theta}^2$ such that $N^1=\{1\} $ and $N^2=\{n\}$, where $N^1$ and $N^2$ are the sets of needy agents in states $\bm{\theta}^1$ and $\bm{\theta}^2$ respectively. 
 Then, DSIC and efficiency imply that  there are truthful dominant strategy message profiles $\bm{\bar{m}}^1$ and 
$\bm{\bar{m}}^2$  such that $g_1(\bm{\bar{m}}^1) =1$ and $g_n(\bm{\bar{m}}^2) =1$.

Consider a sequence of message profiles $\{\bm{m}^r\}_{r=0}^n$ such that $\bm{m}^0= \bm{\bar{m}}^1$, and  for $r =1, \ldots, n$
\begin{itemize}
\item $m^r_r =\bar{m}^2_r$
\item $m^r_k = m^{r-1}_k \mbox { for all } k\neq r.$
\end{itemize}
That is, at stage $r$, only agent $r$ switches from reporting message $\bar{m}^1_r$ to reporting $\bar{m}^2_r$, while all others report what they reported in stage $r-1$. Notice that $\bm{m}^n=\bm{\bar{m}}^2$.

Consider agent 1. Since $\bar{m}^2_1$ is a dominant strategy in state $\bm{\theta}^2$, a unilateral deviation to $\bar{m}^1_1$ cannot affect $1$'s own probability given any report profile of others. Hence,
\begin{equation}\label{th5:eq1} g_1(\bar{m}^2_1, \bar{m}^1_{-1})=1\end{equation}

Take any $r \in X_1$.   Along the sequence  $\{\bm{m}^r\}_{r=0}^n$, $m^r_r$ is a unilateral deviation of $r$ from her dominant strategy $\bar{m}^1_r$ in state $\bm{\theta}^1$. So, she cannot affect either her own probability of being chosen or that of the sum of her friends' probabilities of being chosen. Repeated use of this argument as well as ~\eqref{th5:eq1} establishes that
\begin{equation}\label{th5:eq2}
\sum_{i \in X_1} g_i(\bar{m}^2_1,\dots,  \bar{m}^2_l, \bar{m}^1_{-X^1}) =1\end{equation}

Similarly, no agent $i$  in $X_2$ can  by a unilateral deviation from $\bar{m}^1_i$ to $\bar{m}^2_i$
affect the sum of the probabilities of their enemies being chosen, since $\bar{m}^1_i$ is a dominant strategy in state $\bm{\theta}^1$. So, starting from agent $l+1$(the first agent in $X_2$), and applying this principle repeatedly,  ~\eqref{th5:eq2} gives us 
\[\sum_{i \in X_1} g_i(\bar{m}^2_{X}, \bar{m}^1_{-X}) =1\]
That is,
\[\sum_{i \in Y} g_i(\bar{m}^2_{X}, \bar{m}^1_{-X}) =0\]

But, clearly, we could have started from $g_n(\bm{\bar{m}}^2) =1$, and defined a sequence of changes in order to arrive at 
\[\sum_{i \in Y_2} g_i(\bar{m}^2_{X}, \bar{m}^1_{-X}) =1\]
This contradiction establishes the theorem.
\end{proof}

\subsection{The planner does not know the network}

If we restrict attention to networks satisfying structural balance, the duples mechanism strictly dominates the constant mechanism since it chooses a needy agent with probability greater than  $q$.  On the other hand, there will be networks satisfying structural balance where the 
random dictator chooses a needy agent with probability $q$.  For instance, structural balance is satisfied  if $n$ is even and each agent has exactly one friend and no enemies.  

\begin{proposition}
\label{proposition :duples:sbalance}
On any network $G$ with $n \geq 3$ nodes that satisfies the structural balance property, the duples mechanism chooses a needy agent with probability greater or equal to
\begin{equation*}
q + q(1 - q)\frac{3n - 8}{8(n - 1)}.
\end{equation*}

\end{proposition}

\begin{proof}

 Take any network $G$ that satisfies structural balance with $n \geq 3$ agents. Take any two agents, $i$ and $j$, such that $i \neq j$. Notice that, under the duples mechanism, the remaining agents would vote as follows:
\begin{itemize}
\item Agents in $F_i \cap F_j$, $I_i \cap I_j$, and $E_i \cap E_j$ vote for a needy agent.
\item Agents in $(F_i\setminus (F_j \cup \{j\})) \cup I_i \cap E_j$ vote for $i$, regardless of the neediness of the agents.
\item Agents in $(F_j\setminus (F_i \cup \{i\})) \cup I_j \cap E_i$ vote for $j$, regardless of the neediness of the agents.
\end{itemize}
The three cases above cover all the agents in $V\setminus \{i,j\}$.
Given the observation above, the probability of choosing a needy agent in graph $G$ from among agents $i$ and $j$ is equal to $1-(1-q)^2$ if
\begin{equation}
\label{eq:duples:sbalance:1}
|F_i\setminus (F_j \cup \{j\})| + |I_i \cap E_j| < \frac{n-2}{2} \textrm{ and } |F_j\setminus (F_i \cup \{i\})| + |I_i \cap E_j| < \frac{n-2}{2},
\end{equation}
and it is at least $q$, otherwise.
To see that this is the case, notice that what matters is which of the agents is chosen when exactly one of them is needy. Suppose this is the case.
If both the inequalities in~\eqref{eq:duples:1} are strict then the mechanism chooses the needy agent. Thus the probability of choosing the needy agent is equal to $1 - (1-q)^2 = q + q(1-q)$ in this case. On the other hand, in any case the mechanism chooses a needy agent with probability at least $q$.

If all the F-components in the network are of size smaller than $n/2$ and the sum of sizes of any two F-components which are not parts of the same EF-component is less than $n/2$, the duples mechanism chooses a needy agent with probability equal to $1-(1-q)^2$. To see why, suppose that the mechanism picked two agents, $i$ and $j$ for the remaining nodes to vote on. If both these agents are friends then they belong to the same F component. By the structural balance property, $F_i\setminus \{j\} = F_j\setminus \{i\}$, $E_i = E_j$, and $I_i = I_j$. Hence every agent in $V\setminus \{i,j\}$ is either a common friend, a common enemy or a common impartial of $i$ and $j$. Hence every agent in $V\setminus \{i,j\}$ votes for a needy agent from among $i$ and $j$.
If $i$ and $j$ are enemies then, by the structural balance property, $E_i\setminus \{j\} = F_j$, $E_j\setminus \{i\} = F_i$, and $I_i = I_j$, and $F_i \cap F_j = \varnothing$. Hence all the common impartials of $i$ and $j$ vote for a needy agent from among $i$ and $j$. All the agents in $F_i$ vote for agent $i$, regardless of the neediness statuses of $i$ and $j$ and, similarly, all the agents in $F_j$ always vote for agent $j$. Since each F-component has size less than $n/2$ so $|F_i| < n/2-1$ and $|F_j| < n/2-1$. This, together with $F_i \cap F_j = \varnothing$ and $I_i \cap E_j = I_j \cap E_i = \varnothing$ implies that condition~\eqref{eq:duples:sbalance:1} is satisfied for $i$ and $j$. Thus the duples mechanism chooses a needy agent from among $i$ and $j$ with probability equal to $1-(1-q)^2$.
Lastly, if $i$ and $j$ are impartial, then, by structural balance, $F_i\cap F_j = F_i \cap E_j = E_i \cap F_j = \varnothing$. Hence, $E_i \subseteq I_j$
and $E_j \subseteq I_i$. Moreover, the set of enemies of $i$ forms an F-component that is not a part of an EF-component of $i$ and the set of enemies of $j$ forms an F-component that is not a part of an EF-component of $j$. Thus $F_i\setminus (F_j \cup \{j\}) = F_i$, $I_i \cap E_j = E_j$, $F_j\setminus (F_i \cup \{i\}) = F_j$, and $I_j \cap E_i = E_i$. Since the sum of sizes of any two F-components which are not parts of the same EF-component is less than $n/2$ so
\begin{align*}
|F_i\setminus (F_j \cup \{j\})| + |I_i \cap E_j| & = |F_i| + |E_j| < n/2-1 \textrm{ and}\\
|F_j\setminus (F_i \cup \{i\})| + |I_j \cap E_i| & = |F_j| + |E_i| < n/2-1
\end{align*}
and so condition~\eqref{eq:duples:sbalance:1} is satisfied for $i$ and $j$. Thus the duples mechanism chooses a needy agent from among $i$ and $j$ with probability equal to $1-(1-q)^2$.

Suppose that there exists an F-components in the network of size greater or equal to $n/2$. Call it $A$ and call the other F-component in the EF-component of $A$, $B$. Let $a = |A|$ and $b = |B|$, so that $a \geq n/2$. By the analysis above, the duples mechanism selects a needy agent with probability less than $1-(1-q)^2$ if it picks one agent $i$ in $C$ and an agent $j$ in $B$ ($i$ and $j$ are enemies in this case), or one agent in $A \cup B$ and another agent in $V\setminus (A\cup B)$ ($i$ and $j$ are impartial in this case). In all these cases the probability of choosing a needy agent is at least $q$ and in the remaining cases the probability of choosing a needy agent is $1 - (1-q)^2$. Hence the probability of choosing a needy agent is at least
\begin{align*}
& \left(\frac{2ab}{n(n-1)} + \frac{2(a+b)(n-a-b)}{n(n-1)}\right)q \\
& \qquad \qquad \qquad {} + \left(1-\frac{2ab + 2(a+b)(n-a-b)}{n(n-1)}\right) (1 - (1-q)^2)\\
& \qquad = q + \left(1-\frac{2((a+b)(n-b)-a^2)}{n(n-1)}\right) q(1-q)\\
& \qquad \geq q + \left(1-\frac{n^2 - 4b^2 + 2bn}{2n(n-1)}\right) q(1-q)\\
& \qquad \geq q + \frac{3n - 8}{8(n - 1)}q(1-q)
\end{align*}
where the first inequality follows because the expression increases in $a$ for $a \geq n/2$ and the second inequality follows because the expression is maximised when $b = n/4$.

Suppose all the F-components are of size less than $n/2$ and there exist two F-components in different EF-component of the total size greater or equal to $n/2$. Then there can be at most four distinct F-components, $A$, $B$, $C$, and $D$, such that $|A\cup B| \geq n/2$, $|C\cup D| \geq n/2$, 
$A$ and $B$ are in different EF-components and $C$ and $D$ are in different EF-components. Suppose also that $C$ and $B$ are in different EF-components and $D$ and $A$ are in different EF-components (this is without loss of generality, because if $C$ was in the same EF-component as $B$ or $D$ was in the same EF-component as $A$, we could swap the symbols denoting components $C$ and $D$).
Let $a = |A|$, $b = |B|$, $c = |C|$, and $d = |D|$. 
Suppose that there are four such components. Then it must be that $a + b = n/2$ and $c + d = n/2$.
By the analysis above, the duples mechanism selects a needy agent with probability less than $1-(1-q)^2$ if it picks two impartial agents, $i$ and $j$, such that either $|F_i \cup E_j| \geq n/2$ or $|F_j \cup E_i| \geq n/2$. Thus the probability of choosing such two agents is maximised when agents in $C$
are in the same EF-component as $A$ and the agents in $D$ are in the same EF component as $B$. If $a = b = c = d = n/4$ then a needy agent is picked with probability less than $1-(1-q)^2$ if one agent is selected from the EF-component $A\cup C$ and another one from the EF-component $B\cup D$. The probability of selecting a needy agent is then equal to 
\begin{align*}
& \frac{2(a+c)(b+d)}{n(n-1)}q + \left(1-\frac{2(a+c)(b+d)}{n(n-1)}\right) (1 - (1-q)^2)\\
& \qquad = q + \left(1-\frac{2(a+c)(b+d)}{n(n-1)}\right) q(1-q)\\
& \qquad \geq q + \frac{n - 2}{2(n - 1)} q(1-q) \geq q + \frac{3n - 8}{8(n - 1)}q(1-q).
\end{align*}
A similar analysis shows that in all the remaining cases the probability of choosing a needy agent when all the F-components are of size less than $n/2$ and there exist two F-components in different EF-component of the total size greater or equal to $n/2$, is greater than $q + (3n - 8)q(1-q)/(8(n - 1))$.\end{proof}


\begin{thebibliography}{10}

\bibitem{adachi2014natural}
Tsuyoshi Adachi.
\newblock A natural mechanism for eliciting rankings when jurors have
  favorites.
\newblock {\em Games and Economic Behavior}, 87:508--518, 2014.

\bibitem{alon2011sum}
Noga Alon, Felix Fischer, Ariel Procaccia, and Moshe Tennenholtz.
\newblock Sum of us: Strategyproof selection from the selectors.
\newblock In {\em Proceedings of the 13th Conference on Theoretical Aspects of
  Rationality and Knowledge}, pages 101--110. ACM, 2011.

\bibitem{amoros2009eliciting}
Pablo Amor{\'o}s.
\newblock Eliciting socially optimal rankings from unfair jurors.
\newblock {\em Journal of Economic Theory}, 144(3):1211--1226, 2009.

\bibitem{amoros2011natural}
Pablo Amor{\'o}s.
\newblock A natural mechanism to choose the deserving winner when the jury is
  made up of all contestants.
\newblock {\em Economics Letters}, 110(3):241--244, 2011.

\bibitem{amoros2013picking}
Pablo Amor{\'o}s.
\newblock Picking the winners.
\newblock {\em International Journal of Game Theory}, 42:845--865, 2013.

\bibitem{amoros2016subgame}
Pablo Amor{\'o}s.
\newblock Subgame perfect implementation of the deserving winner of a
  competition with natural mechanisms.
\newblock {\em Mathematical Social Sciences}, 83:44--57, 2016.

\bibitem{amoros2025eliciting}
Pablo Amor{\'o}s.
\newblock Eliciting the deserving winner in the presence of enemies.
\newblock {\em Social Choice and Welfare}, pages 1--31, 2025.

\bibitem{amoros2002scholarship}
Pablo Amor{\'o}s, Luis~C Corch{\'o}n, and Bernardo Moreno.
\newblock The scholarship assignment problem.
\newblock {\em Games and Economic behavior}, 38(1):1--18, 2002.

\bibitem{aziz2019strategyproof}
Haris Aziz, Omer Lev, Nicholas Mattei, Jeffrey Rosenschein, and Toby Walsh.
\newblock Strategyproof peer selection using randomization, partitioning, and
  apportionment.
\newblock {\em Artificial Intelligence}, 275:295--309, 2019.

\bibitem{aziz2016strategyproof}
Haris Aziz, Omer Lev, Nicholas Mattei, Jeffrey~S Rosenschein, and Toby Walsh.
\newblock Strategyproof peer selection: Mechanisms, analyses, and experiments.
\newblock In {\em AAAI}, pages 397--403, 2016.

\bibitem{babichenko2020incentive}
Yakov Babichenko, Oren Dean, and Moshe Tennenholtz.
\newblock Incentive-compatible classification.
\newblock {\em Proceedings of the AAAI Conference on Artificial Intelligence},
  34(05):7055--7062, April 2020.

\bibitem{baumann2018identifying}
Leonie Baumann.
\newblock {\em Identifying the best agent in a network}.
\newblock SSRN, 2018.

\bibitem{ben2014optimal}
Elchanan Ben-Porath, Eddie Dekel, and Barton~L Lipman.
\newblock Optimal allocation with costly verification.
\newblock {\em American Economic Review}, 104(12):3779--3813, 2014.

\bibitem{berga2014impartial}
D~Berga and Riste Gjorgjiev.
\newblock Impartial social rankings.
\newblock {\em Manuscript}, page~8, 2014.

\bibitem{bjelde2017impartial}
Antje Bjelde, Felix Fischer, and Max Klimm.
\newblock Impartial selection and the power of up to two choices.
\newblock {\em ACM Transactions on Economics and Computation}, 5(4), 2017.

\bibitem{bloch2023selecting}
Francis Bloch, Bhaskar Dutta, and Marcin Dziubi{\'n}ski.
\newblock Selecting a winner with external referees.
\newblock {\em Journal of Economic Theory}, 211:105687, 2023.

\bibitem{bloch2022friend}
Francis Bloch and Matthew Olckers.
\newblock Friend-based ranking.
\newblock {\em American Economic Journal: Microeconomics}, 14(2):176--214,
  2022.

\bibitem{bousquet2014near}
Nicolas Bousquet, Sergey Norin, and Adrian Vetta.
\newblock A near-optimal mechanism for impartial selection.
\newblock In Tie-Yan Liu, Qi~Qi, and Yinyu Ye, editors, {\em Web and Internet
  Economics}, pages 133--146, Cham, 2014. Springer International Publishing.

\bibitem{cartwright1956structural}
Dorwin Cartwright and Frank Harary.
\newblock Structural balance: a generalization of heider's theory.
\newblock {\em Psychological review}, 63(5):277, 1956.

\bibitem{fischer2015optimal}
Felix Fischer and Max Klimm.
\newblock Optimal impartial selection.
\newblock {\em {SIAM} Journal on Computing}, 44(5):1263--1285, 2015.

\bibitem{gibbard1977manipulation}
Allan Gibbard.
\newblock Manipulation of schemes that mix voting with chance.
\newblock {\em Econometrica: Journal of the Econometric Society}, pages
  665--681, 1977.

\bibitem{heider1946attitudes}
Fritz Heider.
\newblock Attitudes and cognitive organization.
\newblock {\em The Journal of psychology}, 21(1):107--112, 1946.

\bibitem{holzman2013impartial}
Ron Holzman and Herv{\'e} Moulin.
\newblock Impartial nominations for a prize.
\newblock {\em Econometrica}, 81(1):173--196, 2013.

\bibitem{kattwinkel2023costless}
Deniz Kattwinkel and Jan Knoepfle.
\newblock Costless information and costly verification: A case for
  transparency.
\newblock {\em Journal of Political Economy}, 131(2):504--548, 2023.

\bibitem{kurokawa2015impartial}
David Kurokawa, Omer Lev, Jamie Morgenstern, and Ariel~D Procaccia.
\newblock Impartial peer review.
\newblock In {\em IJCAI}, pages 582--588, 2015.

\bibitem{mackenzie2019axiomatic}
Andrew Mackenzie.
\newblock An axiomatic analysis of the papal conclave.
\newblock {\em Economic Theory}, pages 1--31, 2019.

\bibitem{mylovanov2017optimal}
Tymofiy Mylovanov and Andriy Zapechelnyuk.
\newblock Optimal allocation with ex post verification and limited penalties.
\newblock {\em American Economic Review}, 107(9):2666--94, 2017.

\bibitem{pereyra2023optimal}
Juan~S Pereyra and Francisco Silva.
\newblock Optimal assignment mechanisms with imperfect verification.
\newblock {\em Theoretical Economics}, 18(2):793--836, 2023.

\bibitem{tamura2014impartial}
Shohei Tamura and Shinji Ohseto.
\newblock Impartial nomination correspondences.
\newblock {\em Social Choice and Welfare}, 43(1):47--54, 2014.

\bibitem{yadav2016selecting}
Sonal Yadav.
\newblock Selecting winners with partially honest jurors.
\newblock {\em Mathematical Social Sciences}, 83:35--43, 2016.

\end{thebibliography}

\end{document}